\documentclass[journal]{IEEEtran}

\usepackage{amsfonts,amsmath,amssymb}
\usepackage[english]{babel}
\usepackage[normalem]{ulem}
\usepackage{hyperref}

\usepackage{amsthm}
\usepackage{graphicx}
\usepackage{subcaption,caption}
\usepackage{epsfig}
\usepackage{ulem}
\usepackage{tikz}
\usepackage{mathrsfs}
\usepackage{bbold}             % This package was used to write the one with double bars 
\usepackage{bm}
\usetikzlibrary{shapes,arrows}
\usetikzlibrary{matrix}
\usetikzlibrary{calc,patterns,angles,quotes,arrows,automata}
\graphicspath{ {./imgs/}}

\usepackage[linesnumbered,noend]{algorithm2e}
\RestyleAlgo{ruled}
\SetKwInOut{Input}{Input}
\SetKwInOut{Output}{Output}
\SetKwInOut{Parameters}{Parameters}
\SetKwIF{If}{ElseIf}{Else}{if}{:}{else if}{else}{endif}

\theoremstyle{definition}
\newtheorem{example}{Example}
\usepackage{mathabx}

\renewcommand{\bar}{\overline}

\theoremstyle{theorem}
\newtheorem{problem}{Problem}

\newtheorem{proposition}{Proposition}

\newtheorem{theorem}{Theorem}

\newtheorem{definition}{Definition}

\theoremstyle{definition}

\theoremstyle{remark}
\newtheorem{remark}{Remark}

\usepackage[numbers,sort&compress]{natbib}
\newcommand{\Oc}{\mathcal{O}}

\newcommand{\Rc}{\mathcal{R}}

\newcommand{\Nc}{\mathcal{N}}

\newcommand{\Vc}{\mathcal{V}}

\newcommand{\As}{\mathscr{A}}

\newcommand{\Rb}{\mathbb{R}}

\newcommand{\Rbar}{\bar{R}}

\newcommand{\one}{\mathbb{1}}

\newcommand{\supp}{{\rm supp}}

\newcommand{\alg}{{\rm alg}}

\DeclareMathOperator{\rank}{rank}

\definecolor{cobalt}{rgb}{0.02, 0.29, 0.65}

\definecolor{cobalt}{rgb}{0.0, 0.28, 0.67}

\begin{document}

\title{Robust, positive and exact model reduction\\ via monotone matrices}

\author{Marco Cortese, Tommaso Grigoletto, Francesco Ticozzi, Augusto Ferrante \thanks{The authors are with the Department of Information Engineering, University of Padova, Via Gradenigo 6, 35131 Padova, Italy. Emails: \href{mailto:marco.cortese.2@phd.unipd.it}{\texttt{marco.cortese.2@phd.unipd.it}},
\href{mailto:tommaso.grigoletto@unipd.it}{\texttt{tommaso.grigoletto@unipd.it}},
\href{mailto:ticozzi@dei.unipd.it}{\texttt{ticozzi@dei.unipd.it}}, \href{mailto:augusto@dei.unipd.it}{\texttt{augusto@dei.unipd.it}}.\\ Acknowledgments: The work of F. T. was supported by European Union through NextGenerationEU, within the National Center for HPC, Big Data and Quantum Computing under Projects CN00000013, CN 1, and Spoke 10.}}

\date{\today}

\maketitle

\begin{abstract}
This work focuses on the problem of exact model reduction of positive linear systems, by leveraging minimal realization theory. While determining the existence of a positive reachable realization remains in general an open problem, we are able to fully characterize the cases in which the new model is obtained with non-negative reduction matrices, and hence
positivity of the reduced model is robust with respect to small perturbations of the original system.
The characterization is obtained by specializing monotone matrix theory to positive matrices. In addition, we provide a systematic method to construct positive reductions also when minimal ones are not available, by exploiting algebraic techniques. 
\end{abstract}


\section{Introduction}
Positive dynamical systems are ubiquitous in the modelling of natural and artificial phenomena, ranging from biological systems to transport problems \cite{farina2000positive,rabiner1989}. Despite the substantial amount of research that has been devoted to their analysis and control, a few key open problems remain, even in the linear case \cite{benvenuti2004tutorial,benvenuti2022minimal, ohta}.
For example, it is not known when a positive system admits a minimal realization that is also positive, with the notable exception of {\em symmetric systems} in the single-input single-output (SISO) case, where the system and its dual are the same (see \cite{Grussler}, Theorem 4). In this work, we focus on the existence and the construction of exact reduced-order {\em positive}  
realizations for a given positive linear systems in state-space representation, neither necessarily symmetric nor SISO. 
 Here, we propose and solve a related sub-problem, relevant to model reduction for positive systems:  we seek to find, when possible, an {\em exact} minimal positive reachable realization that can be obtained in a robust way - namely, by restricting the original dynamics to the reachable subspace using positive maps. %\cite{antoulas2005approximation}. 

This extra requirement ensures that the positivity of the reduced-order system is maintained if the original system is affected by sufficiently small perturbations that maintain its positivity, yielding better numerical properties. This scenario is relevant, for example, in networked positive systems associated to a graph topology that structurally prevents controllability \cite{bullo}. A reduction based on a nominal model with positive maps ensures that positivity is maintained for all other models in the same class.
In general, it is a sufficient but not necessary condition for the existence of a positive reachable realization. 

Similar problems have been studied in the literature, where the reduction is obtained by leveraging the network's structure \cite{cheng2021model} or by approximating the system's behaviour \cite{amjad_generalized_2020, kotsalis_balanced_2008, deng_optimal_2011, li2011positivity, virnik, Grussler, Rantzer, virnik2, Feng}. However, our setting is distinctively different: we focus on {\em exact} positive model reduction. Our analysis relies on realization theory and the specialization of the theory of {monotone matrices} to non-negative matrices, which we are able to characterize. The main result then exploits the structure of {\em monotone non-negative matrices} to provide a (computable) test on whether a given projection admits a non-negative full-rank factorization. When one considers the projection onto the reachable subspace, this is shown to be equivalent to the existence of a robust positive reachable realization, that is, a reduced model able to {\em exactly} reproduce the dynamics of the original system. 

When a minimal positive reachable reduction does not exist, we propose a way to relax the minimality requirement and still obtain a reduced positive model, extending an algebraic approach recently developed for hidden-Markov models \cite{tac2023}.
We use the monotonicity-based results to show that, albeit effective, the algebraic extension might in general be nonoptimal, addressing a previously open question. 
Section \ref{sec:monotone} is devoted to developing the necessary results for non-negative monotone matrices, while the following ones define the robust positive reduction problem, characterize its minimal solutions, and propose an algebraic approach to find non-minimal ones.

\subsection{Notation}
Given a vector $x\in\Rb^n$ we denote by $[x]_i$ its $i$-th entry.
We denote by $\one_n$ the column vector with $n$ elements all of which are equal to $1$; If there is no risk of confusion, we drop the subscript and use the simple notation $\one$.  The identity matrix of dimension $n$ is denoted by $I_n$. Given $V \in \Rb^{n \times m}$, we denote by $V^T$ the transpose of $V$. For a square matrix $V \in \Rb^{n \times n}$, its inverse is $V^{-1}$. For a rectangular matrix $V \in \Rb^{n \times m}$ with $n \geq m$ we denote by  $V^\dag$ any matrix that belongs in the set  $V^{-l} := \{ V^\dag | V^\dag V = I_m \}$, defined as the set of the left-inverses of $V$. We denote by the term {\em row sub-matrix of $V$} a matrix whose rows are a subset of the rows of $V$, regardless of the order. Given a vector $x \in \Rb^n$ and the standard basis $\{e_i\}$ for $\Rb^n$, we define the support of $x$ to be the vector space $\mathrm{supp}(x) = \mathrm{span} \{e_i | e_i^T x \neq 0\}$. We denote the image of $V$ by $\mathrm{Im}(V)$. A matrix $V$ (or a vector $v$) is said to be non-negative if all of its entries are greater than or equal to zero.  In this case we write $V \geq 0$ (or $v \geq 0$). To denote a symmetric positive definite (semidefinite) matrix we use the symbol $\succ$ ($\succeq$). We say that a subspace $\Vc \subseteq \Rb^n$ is {\em positively generated} if it admits a basis of non-negative vectors. We denote by $\langle v, w \rangle = v^T w$ the standard inner product, and $\langle v, w \rangle_Q = v^T Q w$ the inner product associated with a square positive definite matrix $Q$.  We denote by $\Rb_+$ the set of all scalars $\lambda \in \Rb$ such that $\lambda \geq 0$, while by $\Rb^n_+$ we denote the real non-negative orthant, i.e. $\Rb^n_+ = \{x | x \in \Rb^n \;, x\geq 0\}$. With $\mathrm{cone}(V)$ we denote the conical hull of the rows of the matrix $V \in \Rb^{n \times m}$, that is, $\mathrm{cone}(V) = \{x  | x = V^T c, \; c \geq 0\}$.

\section{Monotone matrices}\label{sec:monotone}

\subsection{Preliminaries}

Monotone matrix theory and its applications have been extensively studied, see e.g. \cite{monotone} and \cite{collatz}. In the following, after reviewing their definition and a  key existing result, we shall focus on studying and characterizing {\em non-negative} monotone matrices. 

\begin{definition}[Monotone matrices]
    A real matrix $X \in \Rb^{n \times m}$ is said to be monotone if $Xx \geq 0$ implies  $ x \geq 0$, where $x \in \Rb^m$. 
\end{definition}

It has been shown that these type of matrices admit a non-negative inverse. The following characterization of monotone matrices is provided in \cite{monotone}.
\begin{theorem}
\label{thm: monotone}
    Let $X \in \Rb^{n \times m}$ be a real matrix. The following statements are equivalent:
    \begin{enumerate}
        \item $X$ is monotone;
        \item $X$ admits a non-negative left-inverse, i.e. there exists $X^\dag \in \Rb^{m \times n}$ such that $X^\dag X = I_m$, $X^\dag \geq 0$;
        \item $\mathrm{cone}(X) \supseteq \Rb^m_+$.
    \end{enumerate}
\end{theorem}

Notice that the only interesting case is when $n \geq m$. In the case $n < m$ the conditions are never satisfied and the theorem is trivially true. 

\subsection{Characterization of non-negative monotone matrices}
\label{sect: characterization of monotone matrices}
We next specialize the existing results to the case of non-negative monotone matrices. To this end, it is important to recall a property of orthogonal non-negative vectors: 
\begin{remark}
\label{rem:disjoint_support}
    Two non-negative vectors $v,w \in \Rb^n$ are orthogonal if and only if they have disjoint support, i.e. $\mathrm{supp}(v) \cap \mathrm{supp}(w) = \{0\}$. To verify this we can observe that $v$ and $w$ are orthogonal if and only if $v^Tw = \sum_{i=0}^n [v]_i [w]_i = 0$. Since $[v]_i \geq 0$, $[w]_i \geq 0$ for all $i$, then $[v]_i[w]_i\geq0$ for all $i$ and thus the sum of non-negative elements is zero if and only if they are all zero. Moreover $[v]_i[w]_i=0$ if and only if for every index $i$ either $[v]_i = 0$ or $[w]_i = 0$ or both, i.e. they have disjoint support.
\end{remark}  
This fact, together with Theorem \ref{thm: monotone}, is enough to formalize a characterization for square non-negative monotone matrices whose proof can be found in \cite{berman}. %\cite[Chapter 3]{farina2000positive}.
\begin{proposition}
\label{prop: square monotone}
    Let $X \in \Rb^{n \times n}$ be an invertible non-negative square matrix. Then it is monotone if and only if all of its columns (and rows) are mutually orthogonal.
\end{proposition}

The following result extends Proposition
\ref{prop: square monotone} to rectangular matrices: { while such extension} is not new (see again \cite[Chapter 3]{farina2000positive}), we prove it here in a way that serves as an introduction to the proofs of the results of the next sections.
\begin{proposition}
\label{prop: rect monotone}
    Let $X \in \Rb^{n \times m}$ with $n \geq m$ be a full-rank non-negative matrix. Then $X$ is monotone if and only if it contains a set of $m$ distinct orthogonal rows. 
\end{proposition}
\begin{proof}
($\Rightarrow$) Suppose that $X$ is monotone: by point 3 of Theorem 1, $\mathrm{cone}(X) \supseteq \mathbb{R}^m_+$.  This is equivalent to saying that the $m$ rows of $X$ generate the cone $\mathbb{R}^m_+,$ since all the others $n-m$ rows are linearly dependent and therefore are inside the cone generated by these $m$ rows. Reordering the rows of $X$ so that these $m$ rows are at the top, we obtain a square $m \times m$ monotone matrix. According to Proposition \ref{prop: square monotone} its rows are orthogonal.\\
($\Leftarrow$) Observe that a reordering of the rows of $X$ is equivalent to left-multiply $X$ by a permutation matrix $P$. This operation does not alter the row space, namely the span of the rows of $X$ as $\mathrm{span}(X^\top) = \mathrm{span}(X^\top P^\top)$. More importantly, it implies $\mathrm{cone}(PX) = \mathrm{cone}(X)$. Suppose now that $X$ contains a set of $m$ distinct orthogonal rows. There exists a permutation matrix $P \in \Rb^{n \times n}$ that rearranges the $m$ orthogonal rows at the top such that
    \[
    PX = 
    \left[
    \begin{array}{c}
      X_0 \\ \hline
      X_1
    \end{array}
    \right],
    \]
    where $X_0 \in \Rb^{m \times m}$ is a diagonal non-negative matrix. It follows that $\mathrm{cone}(PX) = \mathrm{cone}(X_0) = \Rb^m_+$, proving that $PX$, and thus $X$, is monotone. 
\end{proof}
This Proposition allows us to derive some useful results regarding non-negative projectors.

\subsection{Non-negative full-rank factorization of projections}
Let us start this subsection by observing that, given a subspace $\Vc \subset \Rb^n$ of dimension $m \leq n$, any projector $\Pi_\Vc$ onto $\Vc$  can be written as $\Pi_\Vc = J J^\dag$, for some matrix $J \in \Rb^{n \times m}$ such that $\mathrm{Im}(J) = \Vc$ and $J^\dag J = I_m$. Furthermore, notice that any other projector $\widetilde{\Pi}_\Vc$ onto  $\Vc$ can be written as $\widetilde{\Pi}_\Vc =J \widetilde{J}^{\dag}$ for some other left inverse $\widetilde{J}^{\dag}\neq J^\dag$.

A nontrivial question then arises: given a vector space $\Vc$ is there a projector $\Pi_\Vc$ that can be factorized into $\Pi_\Vc=J J^\dag$, such that $J,J^\dag \geq 0$? The following result characterizes such subspaces and provides a valid factorization.
\begin{theorem}
\label{thm: pos_dec}
    Let $\Vc \subseteq \Rb^n$ be a subspace of dimension $m$. Then the following two conditions are equivalent: \begin{itemize}
        \item There exists a projector $\Pi_\Vc$ with non-negative factors, i.e. $\Pi_\Vc = J J^\dag$ such that $J,J^\dag\geq0$ and $J^\dag J = I_m$;
        \item For any matrix $V \in \Rb^{n \times m}$ with $\mathrm{Im}(V)=\Vc$, there exists a permutation matrix $P$ such that \[PV = \left[\begin{array}{c}V_0 \\ \hline V_1 \end{array}\right]\] 
    where the sub-matrices $V_0\in{\mathbb R}^{m\times m}$ and $V_1\in{\mathbb R}^{(n-m)\times m}$
are such that $\rank(V_0)=m$ and $V_1V_0^{-1}\geq 0$.
    \end{itemize}
\end{theorem} 
\begin{proof}
($\Leftarrow$)
Suppose that there exist $V_0$ with $\rank(V_0)=m$ and $V_1$, sub-matrices of $V$ of the dimensions given in the statement, and such that $V_1 V_0^{-1} \geq 0$.  In this case, we can right-multiply $PV$ by $V_0^{-1}$ (since $V_0$ is full-rank) to find: 
\begin{equation*}
    P V V_0^{-1} = \left[\begin{array}{c}I_m \\ \hline V_1 V_0^{-1}\end{array}\right].
\end{equation*}
Then $PVV_0^{-1}$ is non-negative, since $V_1 V_0^{-1} \geq 0$, and monotone, by Proposition \ref{prop: rect monotone} and by the fact that the first $m$ rows are mutually orthogonal. Because permutations on the left do not affect non-negativity and the monotone property, we have that $V V_0^{-1}$ is also non-negative and monotone. Then, to construct the two factors, we can take, for example, $J=VV_0^{-1}$ and $J^\dag = \left[\begin{array}{c|c}I_m &0\end{array}\right]P$ which are both non-negative and $\mathrm{Im}(J) = \mathrm{Im}(V)=\Vc$. 

($\Rightarrow$)
Assume now that there exist $\Pi_\Vc = J J^\dag$ such that $J,J^\dag\geq0$. Then, $J$ is monotone by Theorem \ref{thm: monotone}. From Proposition \ref{prop: rect monotone} we then have that, since $J$ is non-negative and monotone, there exist a permutation matrix $P$ such that \[PJ = \left[\begin{array}{c} J_0 \\ \hline J_1\end{array}\right]\] with $J_0\in\Rb^{m\times m}$ and $J_1\in\Rb^{n-m\times m}$ submatrices of $J$ such that $J_1$ is non-negative and $J_0$ is a non-negative diagonal full-rank matrix. This directly implies that, $J_1 J_0^{-1} \geq 0$ and we can thus define $\hat{V}=JJ_0^{-1}$. Any other matrix $V\in\mathbb{R}^{n\times m}$ such that ${\rm Im}(V) = \Vc$ can then be written as $V=\hat{V}T$ for some invertible matrix $T$. Then, when multiplied to the left by $P$ we have \[P V = \left[\begin{array}{c} T \\ \hline J_1J_0^{-1} T\end{array}\right]\]
where $J_1J_0^{-1}TT^{-1} = J_1 J_0^{-1} \geq0$, concluding the proof.
\end{proof}
This theorem provides us with (1) a systematic method to verify whether a given subspace $\Vc$ admits a projector $\Pi_\Vc$ with non-negative factors or not and (2) a way to compute a non-negative left inverse of the factor $J$. This allows us to implement an algorithm to perform this verification, here summarized in Algorithm \ref{algo:existence_nonneg_fact}. Notice that, any choice of the matrix $V$ is equivalent to the aim of assessing the existence of a non-negative factorization $\Pi_{\cal V}=JJ^\dag$. This must be the case, as the existence of the factorization is a property of $\Pi_\Vc$ alone. As for the practical implementation, the computational complexity depends only on the number of permutation matrices $P$ to check. More precisely, given a matrix $V \in \Rb^{n \times m}$, $\textrm{Im}(V)=\Vc$, such number is at most $\frac{n!}{(n-m)!}.$

\begin{algorithm}[ht]
    \caption{Check the existence of a non-negative factorization of a projector onto a subspace}
    \label{algo:existence_nonneg_fact}
    \SetAlgoLined
    Let $\Vc \subseteq \Rb^n$\;
    Pick $V \in \Rb^{n \times m}$ such that $\mathrm{Im}(V) = \Vc$\;
    \For{
    all $m$-permutations of $n$, $P$
    }{
    Compute $PV =  \left[\begin{array}{c}V_0 \\ \hline V_1 \end{array}\right]$, $V_0\in\Rb^{m\times m}$\;

    \If{$\rank(V_0)=m$ and $V_1 V_0^{-1} \geq 0$}{
    Compute $J = V V_0^{-1}$\;
    Compute the non-negative left-inverse $J^\dag = \left[\begin{array}{c|c}I_d &0\end{array}\right]P$\;
    \textbf{Stop:} non-negative factorization found\;
    }
    }
    \textbf{Negative output:} a non-negative factorization of a  projector onto $\Vc$ does not exist\;
\end{algorithm}

To conclude this subsection, we shall discuss positively generated vector spaces.  Notice that the existence of a non-negative factorization of a projector $\Pi_\Vc$ implies that $\Vc$ is positively generated by the columns of $\hat{V}$ as defined in the second half of the proof.

\section{Robust positive model reduction}

In the following section we show how the novel non-negative factorization of projectors can be applied to robustly reduce positive systems without losing the positivity of the system. Consider a general discrete-time system in state-space form:
\begin{equation}
    \Sigma : = \begin{cases}
    x(k+1) = A x(k) + B u(k) \\
    y(k) = C x(k)
    \end{cases}
\end{equation}
where $A \in \Rb^{n \times n}$, $B \in \Rb^{n \times m}$, $C \in \Rb^{p \times n} $ are real matrices. Following \cite[Chapter 2]{farina2000positive}, a discrete-time system is said to be (internally) positive if and only if, for any non-negative input sequence $u(k)$ and non-negative initial state $x(0)$, both the state and output sequences $x(k)$ and $y(k)$ remain non-negative.  As shown in \cite{farina2000positive}, the non-negativity constraint on the state and output sequences is equivalent to the non-negativity constraint on the entries of the system matrices $A,B,C$. From now on, with positive system we mean an internally positive system.
We denote the reachability matrix in its standard formulation as $R = [B \; AB \; \ldots \; A^{n-1} B]$ and we define the reachable space as  $\Rc := \mathrm{Im}(R)$ with $q=\dim(\Rc)$. Note that this represents the space of reachable states with inputs that are not necessarily positive: the set of states reachable with positive inputs is in general more restricted \cite{A9}. Nonetheless, for a positive system, $\Rc$ is positively generated by construction. We say that $\Rbar \in \Rb^{n \times q}$ is a truncated reachability matrix if  
$\mathrm{Im}(\bar{R}) = \Rc$. A common way to choose $\Rbar$ is by selecting $q$ linearly independent columns of $R$. 

We now present different positive reduction problems. We start by defining what we mean by equivalent systems.
\begin{definition}Two systems $\Sigma = (A,B,C),$ $\Sigma' = (A',B',C')$ are (input-output) {\em equivalent} if $CA^kB = C' (A')^k B'$ for every $k \geq 0.$
\end{definition}
\noindent Equivalence thus guarantees that the two systems are different state-space realizations of the same input-output relation, since the $CA^kB$ correspond to the coefficients of the series expansion of the system transfer function in $z^{-k}$.
\begin{problem}[Positive model reduction]
    Consider a positive system $\Sigma = (A,B,C)$ of dimension $n$. Find a positive system $\Sigma_r = (A_r, B_r, C_r)$ of dimension $d < n$ such that $\Sigma$ and $\Sigma_r$ are equivalent.
\end{problem}

A widely accepted technique to perform model reduction is to put the system in the reachable (or observable) form and consider only the reachable (or observable) subsystem. The latter procedure is equivalent to projecting the system onto the reachable space. Indeed, given a projector onto the reachable space $\Pi_\Rc = JJ^\dag$ where $\mathrm{Im}(\Pi_\Rc) = \Rc$, exploiting the fact that the reachable space is the smallest $A$-invariant subspace that contains $\mathrm{Im}(B)$, the Markov's coefficients of a general system $\Sigma=(A,B,C)$ are
\[
\begin{split}
    M_{k-1} &= C A^k B = C A^k \Pi_\Rc B = C \Pi_\Rc A^k \Pi_\Rc B \\
    &=(C J) (J^\dag A^k J) (J^\dag B) = (CJ)(J^\dag A J)^k( J^\dag B),
\end{split}
\]
where $(J^\dag A J)^k = (J^\dag A^kJ)$ since $J J^\dag = \Pi_R$. Hence, the reduced system $\Sigma_r$ with $A_r = J^\dag AJ$, $B_r = J^\dag B$, $C_r = CJ$ is equivalent to the original one. However, the positivity constraints on the reduced system  matrices $J^\dag A J, J^\dag B, CJ$ pose significant challenges, especially in presence of perturbations. In a general context, the only way to assure non-negativity of $J^\dag B$ where $B$ can vary among all non-negative matrices is to impose $J^\dag$ to be non-negative. For such  reason, we focus our attention on a novel problem.

\begin{problem}[Robust positive model reduction]
\label{prob:rpmr}
    Consider a positive system $\Sigma = (A,B,C)$ of dimension $n$. Find, if it exists, a projector $\Pi = J J^\dag$, $J \in \Rb^{n \times q}$, $J^\dag \in \Rb^{q \times n}$, $q < n$ with $J^\dag,J\geq 0$, so that $\Sigma'=(J^\dag A J,J^\dag B, CJ)$ is equivalent to $\Sigma$.
\end{problem}
The non-negativity condition on $J^\dag$ and $J$ is enough to ensure the non-negativity of the reduced system.
With robust we mean that small uncertainties on the original model (that maintain the model positive) would not compromise the positivity of the reduction. On the other hand, if $J^\dag$ or $J$ have some negative entries, it may happen that the perturbed model could lose its positivity property when reduced, whereas the unperturbed model would preserve it.  We provide a simple example for completeness.
\begin{example}
    Consider the positive system
    \[
   A =
   \begin{bmatrix}
       1 & 1+\epsilon & 0 & 0 \\
       1 & 0 & 2 & 0 \\
       0 & 0 & 1 & 2 \\
       0 & 0 & 3 & 1
   \end{bmatrix}
   \qquad
   B = \begin{bmatrix}
       1 \\ 1+\epsilon \\ 0 \\ 0
   \end{bmatrix}
    \]
    with $\epsilon\geq 0$.
    Starting with the un-perturbed model, i.e. $\epsilon = 0$, the reachability matrix and its truncated version are the following
    \[
    R =
    \begin{bmatrix}
       1 & 2 & 3 & 4 \\
       1 & 1 & 2 & 3 \\
       0 & 0 & 0 & 0 \\
       0 & 0 & 0 & 0
   \end{bmatrix},\qquad
   \Rbar = 
    \begin{bmatrix}
        1 & 2 \\ 1 & 1 \\ 0 & 0 \\ 0 & 0
    \end{bmatrix}
    \]
    which have rank equal to 2.
    A possible choice of left-inverse of $\Rbar$ is the Moore-Penrose pseudo-inverse
    \[
    \Rbar^\dag = 
    \begin{bmatrix}
        -1 & 2 & 0 & 0 \\ 1 & -1 & 0 & 0
    \end{bmatrix}.
    \]
     Notice that $\Rbar^\dag$ does not enjoy the non-negative property. However, the reduced system result to be positive:
    \[
    A_r = \Rbar^\dag A \Rbar = 
    \begin{bmatrix}
        0 & 1 \\ 1 & 1
    \end{bmatrix}
    \qquad
    B_r = \Rbar^\dag B = 
    \begin{bmatrix}
        1 \\ 0
    \end{bmatrix}.
    \]
On the other hand, when considering a small perturbation on the model, i.e. $\epsilon = 0.1$, the reduction is not assured to preserve the positivity of the system. In fact, in this case the reduced model would not result in a positive system:
\[
    \Rbar^\dag A^* \Rbar = \begin{bmatrix}
       -0.1 & 0.9 \\ 1.1 & 1.1
   \end{bmatrix}
   \qquad
   \Rbar^\dag B^* = \begin{bmatrix}
       1.2 \\ -0.1
   \end{bmatrix}.
    \]
    \qed
\end{example}
As mentioned before, we exploit the fact that projecting onto the reachable space yields an equivalent reduced system. We refer to this approach as reachable Robust Positive Model Reduction (reachable RPMR). With this aim, we firstly consider projectors onto the reachable space, and secondly we make considerations on the dual approach, namely obtaining a reduction by projecting onto the so called observable space.
Exploiting the mathematical result given in the previous section, we delineate the necessary and sufficient conditions regarding the feasibility of the reachable RPMR approach.  

\begin{theorem}
    \label{thm: rpmr}
    Consider a positive system $\Sigma=(A,B,C)$.  Let $\Rc$ be the reachable space. Then there exists a projector $\Pi_\Rc$ onto the reachable space which can be factorized as $\Pi_\Rc = JJ^\dag$ with $J,J^\dag \geq 0$ if and only for any matrix $V \in \Rb^{n \times m}$ with $\mathrm{Im}(V)=\Rc$, there exists a permutation matrix $P$ such that \[PV = \left[\begin{array}{c}V_0 \\ \hline V_1 \end{array}\right],\]   where the sub-matrices $V_0\in{\mathbb R}^{m\times m}$ and $V_1\in{\mathbb R}^{(n-m)\times m}$
are such that $\rank(V_0)=m$ and $V_1V_0^{-1}\geq 0$. Moreover, such projector solves Problem \ref{prob:rpmr}. 
\end{theorem}
\begin{proof}
    Since $\Rc$ is a positively generated space, satisfying the condition of Theorem \ref{thm: pos_dec} is equivalent to saying that there exists a non-negative full-rank factorization of a projector $\Pi_\Rc=J J^\dag$ such that $\mathrm{Im}(\Pi_\Rc) = \Rc$ and $J,J^\dag \geq 0$. 
\end{proof}
Solving Problem \ref{prob:rpmr} also implies that projecting $\Sigma$ onto the reachable space through such $\Pi_\Rc$ we obtain $\Sigma_r = (J^\dag AJ, J^\dag B, CJ)$ that is a positive reachable model that reproduces the reachable dynamics of $\Sigma$. Therefore, Theorem \ref{thm: rpmr} gives necessary and a sufficient condition to ensure a robust positive reachable reduction.

\subsection{Observable RPMR and minimal positive realization}
\label{sect: observable and minimal realiz}

Another way of performing model reduction is by leveraging the knowledge of the system's output map and thus the \textit{observable subspace}.
However, while the reachable space is uniquely defined, the observable space is only defined as {\em a} complement of the unobservable space and is therefore not uniquely determined. More precisely, given a system $\Sigma = (A,B,C)$ of dimension $n$, the observability matrix is defined as 
\[
O := \left[
    \begin{array}{c}
      C \\ \hline
      CA \\ \hline
      \vdots \\ \hline
      C A^{n-1}
    \end{array}
    \right]
\]
and the non-observable space is defined as $\Nc := \mathrm{ker}(O)$.  Any space of the form $\Oc_Q :=\{ x \; | \; \langle x,y \rangle_Q = 0 \; \forall y \in \Nc\}$ is a legitimate observable subspace as long as the matrix $Q \in \Rb^{n \times n}$ is positive definite. In other words, the observable spaces are all the spaces $\Oc_Q = \mathrm{Im}(Q O^T)$, where $Q = Q^T \succ 0$. In cases where the standard observable space $\Oc_I$ does not permit a RPMR, it is possible that a RPMR becomes attainable by considering a different observable space $\Oc_Q$ for some $Q \neq I$. This consideration allows for a characterization for the existence of a minimal ``observable'' robust positive reduction. 
\begin{proposition}
    Consider a positive system $\Sigma = (A,B,C)$. Suppose it admits a reachable RPMR by projecting onto the reachable space and call the reduced reachable system $\Sigma_R$. The system $\Sigma$ admits a minimal RPMR (which is also a minimal realization) if and only if there exists a matrix $Q \succ 0$ such that there exists a projector $\Pi_{\Oc_Q}$ on the corresponding observable subspace $\Oc_Q$ that admits a non-negative full-rank factorization.
\end{proposition}
\begin{proof}
    By the standard method used to obtain a minimal realization (Chapter 6 of \cite{kailath}), we have that if we reduce $\Sigma$ to a reachable subspace first, and to an observable next, then the reduced system has minimal dimension. If we can do both reductions in a robust, positivity preserving way, such minimal realization is also a RPMR.
\end{proof}
While of theoretical interest, this characterization is hard to test, since it implies assessing the possibility of factorizing the projections on any viable complement to the non-observable subspace.

\section{Algebraic approach to RPMR}
\subsection{Obtaining RPMR by extending the reachable space}

In the last section  we have proposed a simple approach to check necessary and sufficient conditions for the existence of a {\em minimal} RPMR. By "minimal" we mean that the reduced model is reachable, i.e. the entire unreachable subspace is eliminated by the reduction. When these conditions are not satisfied, we would hope to be able to perform a {\em non-minimal} RPMR, namely eliminate at least a part of the unreachable subspace. To this aim, we here summarize and extend the algebraic approach presented in \cite{tac2023}. This method allows us to enclose the reachable space to a bigger subspace that admits a positive reduction. We start by introducing the fundamental mathematical concepts in order to better formalize this methodology.

Let us start by defining a product between vectors: Let $x,y \in \Rb^n$ and let $p \in \Rb^n$ such that $[p]_i>0$ for all $i$, then we define the element-wise product (parametrized by $p$) $\wedge_p:\Rb^n\times\Rb^n\to\Rb^n$ as
\[[x \wedge_p y]_i = \frac{[x]_i [y]_i}{[p]_i}.\]
This product notion allows us to define a distorted algebra.
\begin{definition}
    A $p$-distorted algebra $\As$ over $\Rb^n$ is a vector space that is closed under the element-wise multiplication operation $\wedge_p$.
\end{definition}

Given a vector space $\Vc \subseteq \Rb^n$ we denote by $\alg_p(\Vc)$ the closure to a $p$-distorted algebra generated by $\Vc$, i.e. if $x, y\in\Vc$ then $\alpha x + \beta y \in\alg_p(\Vc)$ for all $\alpha,\beta\in\Rb$ and $x\wedge_p y\in\alg_p(\Vc)$. To compute such a closure one can resort to the following fact, see e.g. \cite{tac2023}: \( \alg_p(\Vc) =  p \wedge_\one \alg_\one(p^{-1} \wedge_\one \Vc)\) where $\alg_\one$is the closure to an algebra w.r.t. the product $\wedge_\one$ and $[p^{-1}]_i = \frac{1}{[p]_i}$.

Given a $p$-distorted algebra $\As \subseteq \Rb^n$ of dimension $q$, there exist a set of non-negative {\em idempotent generators} $\{f_i\}$, $i=1, \ldots, q$, i.e. $\mathrm{span}\{f_i\} = \As$ such that $f_i\wedge_p f_j = f_i \delta_{i,j}$ and $\sum_i f_i = p$. 
If we then construct a matrix $J\in\Rb^{n\times q}$ whose columns are the idempotent generators $f_i$ we obtain a monotone matrix. 

\begin{proposition}
\label{prop:distorted_algebra_monotone}
    Consider a $p$-distorted algebra $\As \subseteq \Rb^n$ of dimension $q \leq n$ and let $\{f_i\}$ be its set of idempotent generators. Let 
    \(
        J = \begin{bmatrix}
            f_1 & \dots & f_q
        \end{bmatrix}.
    \)  Then $J$ is monotone.
\end{proposition}
\begin{proof}
    Taking into account that the idempotent generators are non-negative vectors, $f_i\wedge_p f_j = f_i \delta_{i,j}$ holds if and only if they have disjoint support, hence if and only if they are mutually orthogonal. Then, by Proposition \ref{prop: rect monotone}, $J$ is monotone.
\end{proof}
{A direct consequence of Proposition \ref{prop:distorted_algebra_monotone} is that there exists a non-negative left inverse $J^\dag$ of $J$ so that $\Pi_\As :=JJ^\dag$ is a projector onto $\As$. Hence it is always possible to find non-negative factors of a projector onto a distorted algebra.} 
Furthermore, \cite[Theorem 3]{tac2023} proves that by picking $p\in\Rb^n$ such that $p = \sum_i \lambda_i v_i$ where ${v_i}$ are generators of $\Vc$ and $\lambda_i\neq0$ for all $i$ and such that $\supp(p)=\supp(\Vc)$ and $p>0$ provides the smallest distorted algebra $\alg_p(\Vc)$ that contains the vector space $\Vc$.

Finally, the routine to perform a RPMR by projecting onto the reachable space combining the two approaches is summarized Algorithm \ref{algo:rpmr}. Recall that the same routine can be used to perform a RPMR by projecting onto a legitimate observable space.

\begin{algorithm}[ht]
    \caption{Reachable RPMR}
    \label{algo:rpmr}
    \SetAlgoLined
    Consider a positive linear system $\Sigma = (A,B,C)$\; 
    Compute its reachable space $\Rc \subset \Rb^n$\;
    Check the existence of a non-negative factorization of $\Pi_\Rc$, using Alg. \ref{algo:existence_nonneg_fact}\;
    \If{$\exists \Pi_\Rc = J J^\dag$, with $J,J^\dag \geq 0 $}{
    $\Sigma_r = (J^\dag AJ, J^\dag B, CJ)$\;
    \textbf{Stop}\;
    }
    \Else{
    Let $r_i$ be the generators of $\Rc$\;
    Compute $p = \sum_i \lambda_i r_i$ with $\lambda_i\neq0$ for all $i$ and such that $\supp(p)=\supp(\Vc)$ and $p>0$\;
    Compute $\As = \alg_p(\Rc)$\;
    \If{$\mathrm{dim}(\As) < n $}{
        Compute the matrix $J$ whose columns are the idempotent generators of $\As$\;
        Compute a non-negative left-inverse $J^\dag$\;
        Project onto $\As$: $\Sigma_r = (J^\dag AJ, J^\dag B, CJ)$\;
        \textbf{Stop}\;
    }
    }
    \textbf{Negative output:} RPMR could not be performed on system $\Sigma$.
\end{algorithm}

\subsection{Comparison of the results: monotone matrices versus algebraic approach}
A natural question is whether the algebraic approach is able to provide a minimal RPMR if one exists. So far, we have shown that for any $p$-distorted algebra there always exists a projector that can be non-negatively factorized. We now investigate whether all spaces that admit a projector that can be non-negatively factorized are $p$-distorted algebras.
To develop some intuition, we provide a meaningful example.
\begin{example}
    Consider the linear system with state $x(k)\in\Rb^{4}$ and input $u(k)\in\Rb$ defined by the equation $x(k+1) = A x(k) + B u(k)$ with matrices
\[A=\begin{bmatrix}
    0&\varepsilon&0&0\\\varepsilon&0&0&0\\0&0&1&0\\0&0&0&1 
    
\end{bmatrix}
\qquad B = \begin{bmatrix}
    0 \\ 1 \\ 1 \\ 1
\end{bmatrix}
\]
where $\varepsilon\geq0$ is a parameter.
We can then notice that the reachability matrix is 
\[R = \begin{bmatrix}
    0 & \varepsilon & 0 & \varepsilon^3\\ 
    1 & 0 & \varepsilon^2 & 0\\ 
    1 & 1 & 1 & 1\\ 
    1 & 1 & 1 & 1
\end{bmatrix}. \]
In case $\varepsilon=1$ the dimension of the reachable space is $2$ and the truncated reachability matrix is non-negative and monotone:
\[
\Rbar = \begin{bmatrix}
    1 & 0 \\ 0 & 1 \\ 1 & 1\\ 1 & 1
\end{bmatrix},
\qquad
\Rbar^\dag =  \begin{bmatrix}
    1 & 0 & 0 & 0 \\ 0 & 1 & 0 & 0
\end{bmatrix}.
\]
Under this assumption, the reduced system turns out to be of dimension $2$, described by the reduced matrices:
 \[
 A_r = \Rbar^\dag A \Rbar = \begin{bmatrix}
     0 & 1 \\ 1 & 0
 \end{bmatrix}
 \qquad
 B_r = \Rbar^\dag B = \begin{bmatrix}
     0 \\ 1
 \end{bmatrix}
 \]
  However, closing the reachable space $\Rc$ to an algebra w.r.t. the product $\wedge_\one$ we get 
\[
\mathscr{A} = \mathrm{alg}_\one(\Rc) = \mathrm{span}
\begin{Bmatrix}
    \begin{bmatrix}
        1 \\ 0 \\0\\0
    \end{bmatrix},
    \begin{bmatrix}
        0 \\ 1 \\ 0\\0
    \end{bmatrix},
    \begin{bmatrix}
        0 \\ 0 \\ 1\\1
    \end{bmatrix}
\end{Bmatrix}
\]
that results in the reduced system
 \[
 \Tilde{A}_r =  \begin{bmatrix}
     0 & 1 & 0\\ 1 & 0 & 0 \\ 0 & 0 & 1
 \end{bmatrix}
 \qquad
 \Tilde{B}_r =  \begin{bmatrix}
     0 \\ 1 \\ 1
 \end{bmatrix}
 \]
that has dimension equal to $3$.

On the contrary, in the case $\varepsilon \neq 1$, e.g. $\varepsilon = 2$, the reachable space has dimension 3 and it admits a projector with non-negative factorization $\Pi = J J^\dag$
\[
J = \begin{bmatrix}
    1 & 0 & 0 \\ 0 & 1 & 0 \\ 0 & 0 & 1 \\ 0 & 0 & 1
\end{bmatrix},
\qquad
J^\dag = \begin{bmatrix}
    1 & 0 & 0 & 0 \\0 & 1 & 0 & 0 \\ 0 & 0 & 1 & 0
\end{bmatrix}
\]
In this case (and for any $\varepsilon \neq 1$), the algebra that enclose $\Rc$ is the same as in the case $\varepsilon = 1$, and hence provides a reduced model of dimension equal to the one obtained using the monotone matrices approach.
\end{example}

This example shows that in some cases, even though they are very specific, the monotone matrices approach provides a minimal reduction which is smaller than the reduction offered by the algebraic approach. The algebraic structure constrain the reduction to spaces that exhibit mutual orthogonality between its generators. Conversely, the monotone matrices approach can deal also with spaces that do not exhibit this property.  To offer a complete comparison of the approaches, we characterize when the algebraic approach is able to find a minimal RPMR.
\begin{proposition}
    Consider an unreachable pair $(A,B)$. Suppose that it admits a reachable RPMR, i.e. there exists a non-negative factorization of the projector onto the reachable space $\Pi_\Rc = J J^\dag$. Then the reachable space is a $p$-distorted algebra if and only if $J \in \Rb^{n \times q}$ has orthogonal columns. 
\end{proposition}
\begin{proof}
     Recall that the reachable space is $\Rc= \mathrm{Im}(J)$ and denote by $v_i$ the $i$-th column of $J$.\\
    ($\Rightarrow$)
    If $J$ has orthogonal columns, i.e. the column vectors have disjoint support, $v_i \wedge_p v_j = 0 \; \forall i,j $ for any $p$. It follows that $\Rc$ is closed under the element-wise multiplication and hence is a $p$-distorted algebra, with idempotent generators $\{v_i\}$ and $p = \sum_i v_i$.\\
    ($\Leftarrow$)
    Suppose, by contradiction, that not all columns of $J$ are mutually orthogonal, i.e. there exist $i,j$ such that $v_i \wedge_p v_j \neq 0$ for any $p$ such that $\forall i,$ $[p]_i>0$. Recall that, up to a permutation of the rows, $J$ has the following block structure
    \[
J = \left[
    \begin{array}{c}
      V_0 \\ \hline
      V_1
    \end{array}
    \right],
\]
    where $V_0 \in \Rb^{q \times q}$ is a square invertible matrix with mutually orthogonal columns. This implies that for the indices $i,j$ such that $v_i \wedge_p v_j \neq 0$, $v_i \wedge_p v_j$ has the form
      \[
v_i \wedge_p v_j = \left[
    \begin{array}{c}
      0 \\
      \vdots \\
      0 \\
      \hline
      *
    \end{array}
    \right]
\]
where the upper block of zeros has dimension $q$. However, $v_i \wedge_p v_j \notin \Rc$ since all columns of $J$ have one positive entry among the first $q$ entries. Thus $\Rc$ is not an algebra.
\end{proof}

\section{Conclusion}
In this work, after defining RPMR problems, given a positive state space representation of a linear system having transfer matrix $W(z)$, we provide:  (i) a procedure that verifies the existence and constructs
a positive reduced-order reachable realization of the same transfer function $W(z),$ using positive reduction maps;
  (ii) a method to construct a positive reduction on a larger space using algebraic techniques, when the positive reachable  realization of (i) does not exists. The 
procedure (i) is based on monotone matrix theory, and allows us to characterize when the method (ii) is able to obtain a minimal positive reduction. Method (i) is applicable to both discrete-time and continuous-time systems.
However, for continuous-time systems,  the method currently applies only to the subclass of continuous-time positive systems for which with $A \geq 0$.\footnote{For a continuous-time linear system with parameters $A,B,C$, a necessary and sufficient conditions for the nonnegativity of all state and output trajectories corresponding to non-negative  initial states and non-negative inputs is that that $B$, and $C$ are non-negative and the system matrix $A$ is Metzler \cite{farina2000positive}.}

Our theory can also be adapted, by duality, to obtain positive {\em observable} reductions. In this case, however, the non-uniqueness in the definition of the complement to the non-observable subspace (see \cite{kailath}) allows  for providing only sufficient conditions for the existence of robust positive observable reductions.

We believe that our work represents a new direction, and a step forward, in the development of a complete realization theory for positive linear system.
Further developments include an in-depth investigation of the observable reductions, the extension of approach (i) to the general (Metzler) continuous-time case and application to large scale systems under locality constraints.

Lastly, the possibility of exact model reduction hinges on the original state-space model not
being a minimal realization for its own input-output relation. Our methods essentially build
a minimal, or at least smaller, realization that is also positive. Noisy models, which typically
become formally controllable, are not likely reducible with our methods. For this reason, an important development of these techniques is towards approximate model reduction, where accuracy in the prediction is traded for a system of small dimension.

\nocite{*}
\bibliography{ref}

\end{document}